\documentclass[pdflatex,sn-mathphys-num]{sn-jnl}

\usepackage{graphicx}%
\usepackage{multirow}%
\usepackage{amsmath,amssymb,amsfonts}%
\usepackage{amsthm}%
\usepackage{mathrsfs}%
\usepackage[title]{appendix}%
\usepackage{xcolor}%
\usepackage{textcomp}%
\usepackage{manyfoot}%
\usepackage{booktabs}%
\usepackage{subcaption}
\usepackage{cleveref}

\usepackage{float}
\usepackage{algorithm}%
\usepackage{algorithmicx}%
\usepackage{algpseudocode}%
\usepackage{listings}%
\usepackage{tikz}
\usetikzlibrary{arrows.meta,positioning}
\usetikzlibrary{calc}

\theoremstyle{thmstyleone}%
\newtheorem{theorem}{Theorem}
\newtheorem{proposition}[theorem]{Proposition}%

\theoremstyle{thmstyletwo}%
\theoremstyle{thmstylethree}%
\newtheorem{definition}{Definition}%
\theoremstyle{thmstyleone}
\newtheorem{lemma}[theorem]{Lemma}
\newtheorem{corollary}[theorem]{Corollary}
\theoremstyle{thmstyletwo}
\newcommand{\Estimate}{\mathsf{Estimate}}
\newcommand{\E}{\mathbb{E}}
\newcommand{\Prb}{\mathbb{P}}
\newcommand{\core}{\operatorname{core}_2}
\newcommand{\Bridges}{\operatorname{Bridges}}
\newcommand{\wtO}{\widetilde O}
\newcommand{\wtOmega}{\widetilde\Omega}
\newcommand{\ind}{\mathbf 1}
\newcommand{\depth}{\operatorname{depth}}

\begin{document}

\title[Adaptive Black-Box Exactness Barriers for Nearest-Source Girth Estimation in CONGEST]{Adaptive Black-Box Exactness Barriers for Nearest-Source Girth Estimation in CONGEST}


\author{\fnm{Indraveni} \sur{Chebolu}}\email{indravenik@cdac.in}

\author*{\fnm{Arnab} \sur{Mallick}}\email{arnabm@cdac.in}

\author{\fnm{Ch A S} \sur{Murty}}\email{chasmurty@cdac.in}

\author{\fnm{Seema} \sur{Pangal}}\email{seemap@cdac.in}

\author{\fnm{B S} \sur{Rajpurohit}}\email{rajpurohits@cdac.in}

\author{\fnm{Harmesh} \sur{Rana}}\email{raman.rana45@gmail.com}

\affil{\orgdiv{Information Security Group}, \orgname{Centre for Development of Advanced Computing}, \orgaddress{\street{Hardware Park}, \city{Hyderabad}, \postcode{501510}, \state{Telangana}, \country{India}}}



\abstract{Recent multi-scale nearest-source methods give polynomially sublinear approximations for girth in the CONGEST model. We study exactification by adaptive black-box composition while preserving the same fresh exchangeable source-selection primitive. Our scalar-oracle model exposes the sampled source identities and the scalar estimate from every call, allows arbitrary persistent controller state, adaptive source cardinalities and capacities, adaptive stopping, and an arbitrary final decoder; the internal nearest-source tables remain encapsulated.

We first construct, for infinitely many $n$, a same-size pair $G_t^0,G_t^1$ of maximum-degree-three, logarithmic-diameter graphs whose girths are distinct and both $\Theta(\log n)$. A length-transfer construction makes every bulk source contribute identically on the two graphs. The scalar transcripts can separate the pair only when one of $O(\log n)$ interface sources survives a linear nearest-source rank competition. Coupling the adaptive executions with a conditional permutation-rank bound yields an $\Omega(n/\log n)$ expected retained-source workload requirement for constant exactness probability, even with arbitrary final decoding. For the standard sequential packetized realization, the same scale is an expected-round barrier.

A complementary bridgeless family $\widehat H_t$ shows the same $\Omega(n/\log n)$ direct-retuning barrier on bounded-degree graphs with minimum degree at least two, no bridges, and $2$-core equal to the whole graph. Finally, under a known promise $g\ge h$, one full-source call from $\Theta(n/h)$ uniformly sampled sources computes exact girth with constant probability in $O(n/h+D)$ rounds, matching the $n/g$ source scale.}




\maketitle

\section{Introduction}

The girth of an undirected graph is the length of its shortest cycle. In CONGEST, exact girth has a general $O(n)$-round route, while the best known exact algorithms that exploit the girth value become close to linear as the girth grows~\cite{PelegRodittyTal2012,CensorHillelEtAl2020}. The unrestricted exact problem lies between a $\widetilde\Omega(\sqrt n)$ lower bound inherited from $(2-\varepsilon)$-approximation, for every fixed $\varepsilon>0$, and the linear upper bound~\cite{FrischknechtHolzerWattenhofer2012}. In contrast, recent approximation results are polynomially sublinear. Chechik, Lifshitz, and Mukhtar give, for every constant integer $f>2$, an $f$-approximation to unweighted undirected girth in $\wtO(n^{1/f}+D)$ rounds~\cite{ChechikLifshitzMukhtar2026}. A central ingredient repeatedly samples sources, retains a bounded number of nearest-source BFS records, and invokes the same scalar girth estimator.

This motivates a structural question: \emph{how much exact information can adaptive composition extract from the nearest-source estimator while preserving its source-selection primitive?} We study this question at two complementary levels.

\paragraph{Arbitrary-decoder scalar-oracle barrier.}
Our main model allows a controller to observe both the sampled source set and the scalar estimate after every call. It may retain the complete exposed transcript, choose all future source cardinalities and nearest-source capacities adaptively, stop at an arbitrary transcript-dependent time, and apply an arbitrary final decoder. Thus the result does not rely on returning one of the observed estimates or on hiding sampled source identities.

For $t\ge7$, we construct two graphs $G_t^0,G_t^1$ on the same labeled vertex set, with
\[
 \Delta(G_t^b)\le3,\qquad D(G_t^b)=\Theta(\log n_t),
\]
and distinct girths
\[
 \gamma_t^0=16t+1,\qquad \gamma_t^1=16t+3.
\]
The construction transfers two vertices from attachment stems into the unique cycle. This increases the girth by two while moving every vertex in two exponentially large binary trees one hop closer to the cycle. Since a tree source $s$ contributes ``cycle length plus twice distance to the cycle,'' the two changes cancel exactly. Moreover, the nearest-source rank of every such bulk source at its witness endpoints is identical in the two graphs.

Only an interface set $I_t$ of $|I_t|=\gamma_t^1=O(\log n_t)$ vertices can create different scalar outputs. Every interface source faces at least $\Omega(n_t)$ strictly closer competitors at one of its witness endpoints. Coupling the two executions call by call gives
\begin{equation}\label{eq:intro-coupling}
 \Prb[\text{the exposed transcripts separate}]
 \le \frac{6\gamma_t^1}{n_t}\,\E[W_*],
\end{equation}
where $W_*$ is the retained-source workload accumulated along the common transcript. Exactness with probability at least $2/3$ on both inputs requires the transcripts to separate with probability at least $1/3$, and therefore
\begin{equation}\label{eq:intro-main}
 \E[W]=\Omega(n_t/\gamma_t^1)
 =\Omega(n_t/\log n_t).
\end{equation}
This is an indistinguishability statement for adaptive scalar-oracle composition with arbitrary final decoding.

\paragraph{Bridgeless robustness.}
We also retain a structurally robust family $\widehat H_t$: a bounded-degree odd cactus with minimum degree at least two, no bridges, logarithmic diameter, and $\core(\widehat H_t)=\widehat H_t$. Its unique shortest cycle has length $g_t=8t+1=\Theta(\log\widehat n_t)$. For the direct-retuning subclass, where the final answer is the minimum scalar estimate observed, we prove
\begin{equation}\label{eq:intro-bridge}
 \Prb[\mathcal A(\widehat H_t)=g_t]
 \le \frac{25g_t}{\widehat n_t}\E[W].
\end{equation}
Thus the same $\Omega(n/\log n)$ workload scale persists after $2$-core extraction and bridge elimination.

\paragraph{Matching source scale.}
Under a known promise $g\ge h$, one call with $\Theta(n/h)$ uniformly sampled sources and capacity equal to the sample size computes exact girth with constant probability in $O(n/h+D)$ rounds. Hence the $n/g$ dependence is tight for direct source hitting.

The proofs share one geometric principle. In an odd cactus, each source has exactly one cycle-closing edge per cycle block, and its candidate equals the block length plus twice the source-to-block distance. Exactness is therefore governed by whether the relevant source survives nearest-source competition at that witness edge. The main pair uses this formula to obtain scalar invariance; the bridgeless family uses it to force exact certification through a crowded antipodal edge.

\section{Related Work and Paradigm Context}\label{sec:related}

\paragraph{Exact girth and approximation.}
Peleg, Roditty, and Tal established foundational CONGEST algorithms for diameter and girth~\cite{PelegRodittyTal2012}, and Censor-Hillel et al. later obtained the parameterized exact-girth bound $\wtO(\min\{g\,n^{1-1/\Theta(g)},n\})$ together with faster algorithms for several fixed even-cycle problems~\cite{CensorHillelEtAl2020}. The general exact problem has the $O(n)$ upper bound and a $\widetilde\Omega(\sqrt n)$ lower bound inherited from $(2-\varepsilon)$-approximation~\cite{FrischknechtHolzerWattenhofer2012}.

Manoharan and Ramachandran give the most relevant bridge between the earlier girth literature and recent approximation tradeoffs~\cite{ManoharanRamachandran2024}. For undirected unweighted graphs they improve the $(2-1/g)$-approximation to $\wtO(\sqrt n+D)$ rounds and prove additional lower bounds for larger constant approximation ratios. In the minimum-weight-cycle problem, they prove near-linear lower bounds for exact and $(2-\varepsilon)$-approximate computation in undirected weighted graphs and in directed graphs, together with sublinear $2$- or $(2+\varepsilon)$-approximation algorithms. Chechik, Lifshitz, and Mukhtar obtain an $f$-approximation for unweighted undirected girth in $\wtO(n^{1/f}+D)$ rounds for every constant integer $f>2$~\cite{ChechikLifshitzMukhtar2026}. Their multi-scale nearest-source estimator is the architecture studied here.

\paragraph{Technique-specific barriers.}
Sublinear CONGEST cycle detection has developed through both algorithmic improvements and structural analyses of successful techniques. Eden et al. gave sublinear algorithms for small cliques and even cycles and isolated obstacles encountered by several communication-complexity reductions~\cite{EdenEtAl2019}. Censor-Hillel et al. advanced exact detection for several fixed even cycles~\cite{CensorHillelEtAl2020}, while Gonen--Oshman and Le Gall--Miyamoto give complementary lower bounds for subgraph and induced-cycle detection~\cite{GonenOshman2018,LeGallMiyamoto2021}. A close methodological precedent is Fraigniaud, Luce, and Todinca's SIROCCO work on the local threshold-based colored-BFS paradigm~\cite{FraigniaudLuceTodinca2023}; Fraigniaud, Luce, Magniez, and Todinca subsequently introduced a global-threshold approach for $C_{2k}$-freeness~\cite{FraigniaudLuceMagniezTodinca2024}. Our results similarly quantify the information available through adaptive composition of a specific distributed primitive.

\section{Model and Scalar-Oracle Composition}\label{sec:model}

\paragraph{CONGEST and notation.}
We consider a simple connected undirected unweighted graph $G=(V,E)$ with $n=|V|$ and hop diameter $D$. Vertices have distinct $O(\log n)$-bit identifiers. In each synchronous round each edge carries $O(\log n)$ bits in each direction; local computation is free~\cite{Peleg2000}. For vertices $u,v$, write $d_G(u,v)$ for graph distance, abbreviated to $d(u,v)$. We write $g(G)$ for girth, with $g(G)=+\infty$ for an acyclic graph, and $\Delta(G)$ and $\delta(G)$ for maximum and minimum degree. A \emph{bridge} is an edge whose deletion disconnects the graph, $\Bridges(G)$ is the set of bridges, and $\core(G)$ denotes the $2$-core obtained by repeatedly deleting vertices of degree below two. A \emph{cactus} is a connected graph in which any two simple cycles share at most one vertex. The symbols $\wtO$ and $\wtOmega$ suppress polylogarithmic factors in $n$. ``With high probability'' means probability at least $1-n^{-c}$ for some fixed constant $c>0$.

\paragraph{Nearest-source estimator.}
For $S\subseteq V$ and $k\ge1$, let $U(S,k,v)$ be the $\min\{k,|S|\}$ closest sources to $v$, breaking distance ties by source identifier. A source $s$ is \emph{retained at $v$} when $s\in U(S,k,v)$. Source detection constructs these records, exact distances, and shortest-path parent pointers in $O(k+D)$ rounds in the relevant unbounded-depth setting~\cite{LenzenPeleg2013}.

We use the deterministic estimator $\Estimate(G,S,k)$ of Chechik et al.~\cite{ChechikLifshitzMukhtar2026}. After source detection, neighboring vertices exchange retained records. For an edge $\{x,y\}$ and a source $s$ retained at both endpoints, the edge is a witness for $s$ when neither endpoint is the other's parent in the $s$-rooted BFS tree. It contributes
\begin{equation}\label{eq:witness}
 d(s,x)+d(s,y)+1.
\end{equation}
The estimator returns the minimum witness value, or $+\infty$ if no witness exists. It never returns a value below the true girth. If $|S|=Q$, every vertex retains $B=\min\{Q,k\}$ records. In the standard packetized realization, the table-exchange stage takes $\Theta(B)$ rounds up to absolute constants.

\begin{definition}[Fresh exchangeable scalar-oracle call]\label{def:fresh}
Fix a graph $G$ and a pre-sampling information state $\mathcal G$. The call has fixed a cardinality $Q\in\{0,\ldots,n\}$ and a capacity $k\ge1$, both measurable with respect to $\mathcal G$. Conditional on $\mathcal G$, it draws
\[
 S\sim\operatorname{Unif}\binom{V}{Q}
\]
and exposes the pair $(S,M)$, where $M=\Estimate(G,S,k)$. The internal nearest-source records remain inside the call. Independent Bernoulli sampling is included by first exposing the realized cardinality and then conditioning on it.
\end{definition}

\begin{definition}[Adaptive scalar-oracle exactification]\label{def:oracle}
Let $(\mathcal F_j)_{j\ge0}$ be the controller-visible transcript filtration.
The controller's graph-dependent information is exactly the sequence exposed by the scalar-oracle calls. Its initial information $\mathcal F_0$ consists of graph-independent public information shared by the compared executions, including the common labeled vertex set and its size.

Before call $j$, using $\mathcal F_{j-1}$ and fresh private randomness, the
controller decides whether to stop and, if it continues, fixes $Q_j,k_j$.
Let $\mathcal G_j\supseteq\mathcal F_{j-1}$ be the resulting pre-sampling
sigma-field. The call draws a fresh uniform $Q_j$-subset $S_j$ and exposes
$(S_j,M_j)$ with
\[
 M_j=\Estimate(G,S_j,k_j).
\]
The controller may retain arbitrary state, and its final output is an arbitrary
measurable function of the complete exposed transcript and its private
randomness. For an executed call define
\[
 B_j=\min\{Q_j,k_j\},\qquad W=\sum_{j=1}^{T}B_j,
\]
where $T$ is the stopping time. We assume $T<\infty$ almost surely and
$\E[W]<\infty$ on the inputs under consideration.
\end{definition}

\begin{figure}[H]
\centering
\begin{tikzpicture}[>=Latex, every node/.style={font=\footnotesize}]
\node[draw,rounded corners,align=center,minimum width=3.6cm,minimum height=8mm] (hist) at (-4.6,0) {$\mathcal F_{j-1}$\\complete exposed transcript};
\node[draw,rounded corners,align=center,minimum width=3.6cm,minimum height=8mm] (par) at (0,0) {stop/continue; choose $Q_j,k_j$};
\node[draw,rounded corners,align=center,minimum width=3.6cm,minimum height=8mm] (sample) at (4.6,0) {$S_j$\\fresh uniform $Q_j$-subset};
\node[draw,rounded corners,align=center,minimum width=3.6cm,minimum height=8mm] (est) at (0,-1.55) {$\Estimate(G,S_j,k_j)$\\nearest-source tables internal};
\node[draw,rounded corners,align=center,minimum width=3.6cm,minimum height=8mm] (out) at (4.6,-1.55) {expose $(S_j,M_j)$};
\draw[->] (hist)--(par); \draw[->] (par)--(sample);
\draw[->,rounded corners] (sample.south) -- ++(0,-.35) -| (est.east);
\draw[->] (est)--(out);
\draw[->,rounded corners] (out.south) -- ++(0,-.45) -| (hist.south)
 node[pos=.48,below,align=center] {arbitrary state, parameter adaptation, stopping, and final decoding};
\end{tikzpicture}
\caption{Adaptive scalar-oracle composition. Sampled source identities and scalar estimates are exposed; nearest-source tables stay encapsulated inside each call.}\label{fig:oracle}
\end{figure}
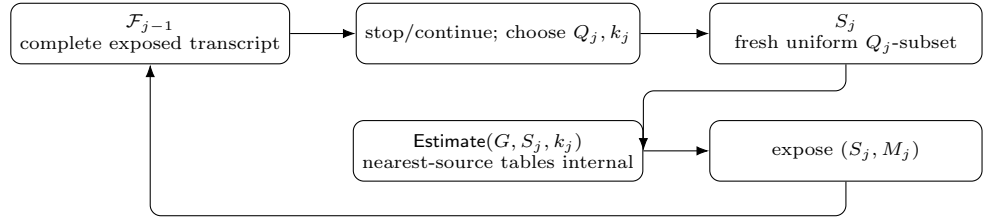

\begin{definition}[Direct retuning]\label{def:direct}
A \emph{direct-retuning controller} is an adaptive scalar-oracle exactification with $T\ge1$ almost surely whose final output is $\min_{1\le j\le T}M_j$. Parameter choices and stopping may use the complete exposed transcript $(S_i,M_i)_{i<j}$.
\end{definition}

\begin{proposition}[Closure of parameter-level retuning]\label{prop:closure}
Consider a method obtained from the multi-scale nearest-source estimator architecture by choosing, possibly adaptively, the number and order of calls, fresh Bernoulli or uniformly fixed-cardinality source samples, per-call capacities, continuation rules, stopping rules, and scalar post-processing based on the exposed pairs $(S_i,M_i)$. Such a method is an adaptive scalar-oracle exactification in the sense of \cref{def:oracle}. If its final answer is the minimum scalar estimate observed, it is a direct-retuning controller in the sense of \cref{def:direct}.
\end{proposition}
\begin{proof}
Conditioned on its realized cardinality, a Bernoulli source sample is uniform among subsets of that cardinality, so every call satisfies \cref{def:fresh}. The transcript filtration records the exposed source identities and scalar estimates and therefore supports arbitrary adaptive choices of subsequent cardinalities, capacities, call order, and stopping. Scalar post-processing is exactly the measurable final decoder permitted by \cref{def:oracle}; taking the minimum of the observed estimates gives \cref{def:direct}.
\end{proof}

\paragraph{Information boundary.}
The scalar-oracle architecture treats source cardinality, capacity, scale order, repetition, stopping, and scalar decoding as controller-level parameters. Graph-dependent source identities enrich source selection; exported distances, parents, or nearest-source tables enrich the oracle output; persistent table entries and joint cross-scale compression enrich the state carried between calls; block decompositions and auxiliary cycle routines enrich the graph-processing layer. These mechanisms enlarge the information interface beyond parameter-level composition of the nearest-source estimator.

\section{Odd-Cactus Witness Geometry}\label{sec:witness}

\begin{lemma}[Unique shortest paths in an odd cactus]\label{lem:unique-paths}
Every pair of vertices in a cactus whose cycle blocks are all odd has a unique shortest path.
\end{lemma}
\begin{proof}
The block-cut tree fixes the sequence of blocks traversed by a simple path. Inside each odd cycle, the two arcs between the entry and exit vertices have different lengths, so exactly one is shorter. Concatenating these local choices yields a unique shortest path.
\end{proof}

For a source $s$ and an odd cycle block $Z$, let $a_Z(s)$ be the gate through which the unique shortest path from $s$ enters $Z$, with $a_Z(s)=s$ when $s\in Z$.

\begin{lemma}[Odd-cycle witness formula]\label{lem:block-witness}
Let $G$ be an odd cactus, fix a source $s$, and let $Z$ be a cycle block of length $|Z|=2r+1$. In the BFS tree rooted at $s$, exactly one edge of $Z$ is not a parent-child edge. Its endpoints are both at $Z$-distance $r$ from $a_Z(s)$, and its estimator candidate is
\begin{equation}\label{eq:block-candidate}
 |Z|+2d(s,a_Z(s)).
\end{equation}
\end{lemma}
\begin{proof}
By \cref{lem:unique-paths}, parent pointers are forced. For $v\in Z$,
\[
 d(s,v)=d(s,a_Z(s))+d_Z(a_Z(s),v).
\]
Exactly one edge of the odd cycle joins the two vertices at distance $r$ from the gate. Every other cycle edge joins consecutive distance layers and is a parent-child edge. Substitution into \cref{eq:witness} gives $2(d(s,a_Z(s))+r)+1$.
\end{proof}

The same formula applies to a retained source in a partial nearest-source table: when the source is retained at both witness endpoints, the reported distances and parent pointers are its exact BFS data.

\begin{lemma}[Uniform-subset rank bound]\label{lem:rank-bound}
Let $z$ be a distinguished vertex and suppose $N$ other vertices are strictly closer than $z$ to a vertex $x$. If $S$ is a uniformly random $Q$-subset of an $n$-vertex ground set, then for every $k\ge1$,
\begin{equation}\label{eq:rank}
 \Prb[z\in S\text{ and }z\in U(S,k,x)]
 \le \min\left\{\frac{Q}{n},\frac{k}{N+1}\right\}.
\end{equation}
\end{lemma}
\begin{proof}
The first bound is $Q/n$. For the second, generate $S$ as the first $Q$ positions of a uniformly random permutation. Among $z$ and its $N$ strictly closer competitors, the relative rank of $z$ is uniform on $\{1,\ldots,N+1\}$. If $z$ is selected and retained at $x$, at most $k-1$ closer competitors can precede it, so its relative rank is at most $k$.
\end{proof}

\section{A Same-Size Scalar-Indistinguishability Pair}\label{sec:pair}

Fix $t\ge7$, let $r=8t$ and $m=2^t-1$, and take two vertex-disjoint perfect binary trees $T_A,T_B$ of height $t-1$, with roots $a,b$. All vertex identifiers below are fixed and shared by the two graphs.

Let
\[
 P=A,p_1,\ldots,p_{r-1},B
\]
be an $A$--$B$ path of length $r$, and let
\[
 Q=A,q_1,\ldots,q_r,B
\]
be an internally disjoint $A$--$B$ path of length $r+1$. Introduce two further vertices $y_A,y_B$.

In $G_t^0$, the unique cycle is $P\cup Q$, and the two trees are attached by the length-two stems
\[
 A-y_A-a,\qquad B-y_B-b.
\]
In $G_t^1$, attach the roots directly by $A-a$ and $B-b$, and place $y_A,y_B$ into the second cycle arc:
\[
 A,y_A,q_1,\ldots,q_r,y_B,B.
\]
Thus two vertices are transferred from the stems to the cycle while the labeled vertex set stays unchanged. Define the interface set
\[
 I_t=V\setminus\bigl(V(T_A)\cup V(T_B)\bigr).
\]

\begin{figure}[H]
\centering
\begin{subfigure}[t]{0.47\textwidth}
\centering
\begin{tikzpicture}[scale=.72,every node/.style={font=\scriptsize}]
\coordinate (A) at (-2.3,0); \coordinate (B) at (2.3,0);
\fill (A) circle (1.5pt) node[left] {$A$}; \fill (B) circle (1.5pt) node[right] {$B$};
\draw[thick] (A) .. controls (-1,1.35) and (1,1.35) .. (B) node[midway,above=6pt] {$P:r$};
\draw[thick] (A) .. controls (-1,-1.2) and (1,-1.2) .. (B) node[midway,below=7pt] {$Q:r+1$};
\coordinate (yA) at (-2.3,-1.0); \coordinate (a) at (-2.3,-1.8);
\coordinate (yB) at (2.3,-1.0); \coordinate (b) at (2.3,-1.8);
\draw (A)--(yA)--(a); \draw (B)--(yB)--(b);
\fill (yA) circle (1.3pt) node[left] {$y_A$}; \fill (yB) circle (1.3pt) node[right] {$y_B$};
\fill (a) circle (1.3pt) node[left] {$a$}; \fill (b) circle (1.3pt) node[right] {$b$};
\draw (a)--(-2.8,-2.45); \draw (a)--(-1.8,-2.45);
\draw (b)--(1.8,-2.45); \draw (b)--(2.8,-2.45);
\node at (0,-2.6) {$T_A\qquad\qquad T_B$};
\end{tikzpicture}
\caption{$G_t^0$: two length-two stems.}
\end{subfigure}\hfill
\begin{subfigure}[t]{0.47\textwidth}
\centering
\begin{tikzpicture}[scale=.72,every node/.style={font=\scriptsize}]
\coordinate (A) at (-2.3,0); \coordinate (B) at (2.3,0);
\fill (A) circle (1.5pt) node[left] {$A$}; \fill (B) circle (1.5pt) node[right] {$B$};
\draw[thick] (A) .. controls (-1,1.35) and (1,1.35) .. (B) node[midway,above=6pt] {$P:r$};
\draw[thick] (A) .. controls (-1,-1.2) and (1,-1.2) .. (B) node[midway,below=7pt] {$r+3$};
\fill (-1.85,-.48) circle (1.3pt) node[below left] {$y_A$};
\fill (1.85,-.48) circle (1.3pt) node[below right] {$y_B$};
\coordinate (a) at (-2.3,-1.55); \coordinate (b) at (2.3,-1.55);
\draw (A)--(a); \draw (B)--(b);
\fill (a) circle (1.3pt) node[left] {$a$}; \fill (b) circle (1.3pt) node[right] {$b$};
\draw (a)--(-2.8,-2.2); \draw (a)--(-1.8,-2.2);
\draw (b)--(1.8,-2.2); \draw (b)--(2.8,-2.2);
\node at (0,-2.35) {$T_A\qquad\qquad T_B$};
\end{tikzpicture}
\caption{$G_t^1$: $y_A,y_B$ move into the cycle.}
\end{subfigure}
\caption{Length transfer. The cycle grows by two while every tree vertex moves one hop closer to it.}\label{fig:pair}
\end{figure}
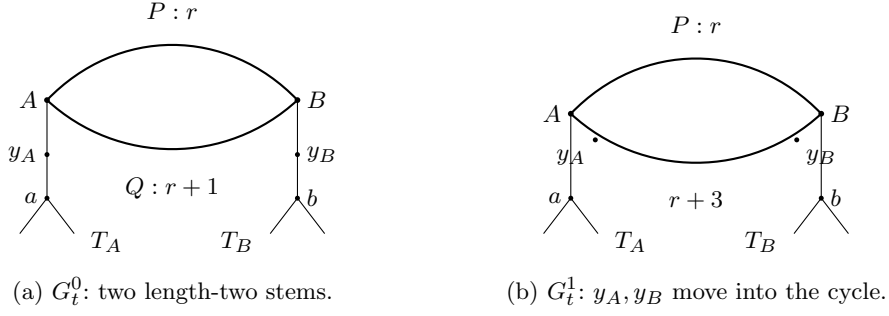

\begin{lemma}[Pair structure]\label{lem:pair-structure}
The graphs $G_t^0,G_t^1$ have the same
\[
 n_t=2^{t+1}+16t+1
\]
vertices and maximum degree at most three. Each is unicyclic, with
\[
 \gamma_t^0=g(G_t^0)=2r+1=16t+1,
 \qquad
 \gamma_t^1=g(G_t^1)=2r+3=16t+3.
\]
Moreover $D(G_t^b)=\Theta(t)=\Theta(\log n_t)$.
\end{lemma}
\begin{proof}
The two trees contribute $2m$ vertices. The remaining labeled set has $2r+3$ vertices in both graphs, giving $n_t=2m+2r+3$. Each graph consists of one cycle with two trees attached, so the displayed cycle is unique and determines the girth. The cycle itself gives diameter at least $r$, while any two vertices can route through the cycle using at most $r+2t+2$ edges. Since $n_t=\Theta(2^t)$, the diameter and both girths are $\Theta(\log n_t)$.
\end{proof}

For a source in a unicyclic odd graph, call its unique non-parent cycle edge its \emph{cycle witness}. A sampled source is \emph{active} when it is retained at both endpoints of this edge.

\begin{lemma}[Bulk-source invariance]\label{lem:bulk-invariance}
Fix a source $s\in T_A\cup T_B$. For every common sampled set $S$ containing $s$ and every capacity $k$, the source $s$ is active in $G_t^0$ if and only if it is active in $G_t^1$. When active, it contributes the same scalar candidate in both graphs.
\end{lemma}
\begin{proof}
By symmetry, take $s\in T_A$ and let $d=\depth_{T_A}(s)$. Its gate to the cycle is $A$. In $G_t^0$ its cycle witness is $\{q_r,B\}$, while in $G_t^1$ it is $\{q_r,y_B\}$. By \cref{lem:block-witness}, its candidate is
\[
 \gamma_t^0+2(d+2)=2r+5+2d
 =\gamma_t^1+2(d+1).
\]

It remains to compare retention. Match the endpoint $q_r$ to itself and $B$ in $G_t^0$ to $y_B$ in $G_t^1$. At either matched endpoint, the vertices preceding $s$ in the distance-then-identifier order are exactly
\begin{equation}\label{eq:predecessors}
 I_t\;\cup\;V(T_B)\;\cup\;
 \{u\in T_A:(\depth(u),\operatorname{ID}(u))
 <_{\mathrm{lex}}(d,\operatorname{ID}(s))\}.
\end{equation}
Indeed, a vertex $u\in T_A$ of depth $e$ is at distance $e+r+2$ from each matched endpoint in the corresponding graph. Every vertex of $T_B$ is at distance at most $t+2$, while every vertex of $I_t$ is at distance at most $r+1$; both are strictly closer than $s$, whose distance is $r+d+2$. Thus the predecessor set in \cref{eq:predecessors} is identical in the two graphs. For a fixed sample $S$, membership of $s$ among the $k$ closest sampled sources is therefore identical at both matched witness endpoints. The argument for $s\in T_B$ is symmetric.
\end{proof}

\begin{lemma}[Interface crowding]\label{lem:interface-crowding}
For every $z\in I_t$ and $b\in\{0,1\}$, one endpoint $x$ of the cycle witness of $z$ in $G_t^b$ has at least $m=2^t-1$ vertices strictly closer to $x$ than $z$. For $t\ge7$,
\[
 m+1\ge n_t/3.
\]
\end{lemma}
\begin{proof}
Every cycle-witness endpoint lies on the unique cycle. Every cycle vertex is within distance at most $\lfloor(r+3)/2\rfloor\le4t+1$ of $A$ or $B$. Choose the nearer anchor $C$. Every vertex of the tree attached at $C$ is within an additional $t+1$ hops in $G_t^0$ and within an additional $t$ hops in $G_t^1$. Hence all $m$ vertices of that tree lie within distance at most $5t+2$ of $x$.

If $z$ lies on the cycle, then its distance to either cycle-witness endpoint is $r$ in $G_t^0$ and $r+1$ in $G_t^1$. The two stem vertices $y_A,y_B$ in $G_t^0$ have witness distance $r+1$. Since $5t+2<8t=r$, all $m$ tree vertices are strictly closer to $x$ than $z$. Finally, $m+1=2^t\ge n_t/3$ for $t\ge7$ because $2^t\ge16t+1$.
\end{proof}

\begin{lemma}[One-call scalar coupling]\label{lem:pair-call}
Couple one fresh scalar-oracle call on $G_t^0$ and $G_t^1$ by using the same uniform $Q$-subset $S$ and the same capacity $k$. Let $B=\min\{Q,k\}$. Then
\begin{equation}\label{eq:pair-call}
 \Prb\!\left[\Estimate(G_t^0,S,k)\ne\Estimate(G_t^1,S,k)\mid Q,k\right]
 \le \min\left\{1,\frac{6\gamma_t^1 B}{n_t}\right\}.
\end{equation}
\end{lemma}
\begin{proof}
By \cref{lem:bulk-invariance}, every bulk source in $T_A\cup T_B$ is active in both graphs simultaneously and contributes the same candidate. Therefore unequal scalar outputs require some interface source $z\in I_t$ to be active in at least one graph.

Fix $z$ and one graph. By \cref{lem:interface-crowding,lem:rank-bound}, activity implies retention at a witness endpoint with at least $m$ strictly closer vertices, so
\[
 \Prb[z\text{ is active}]
 \le \min\left\{\frac{Q}{n_t},\frac{k}{m+1}\right\}
 \le \frac{3B}{n_t}.
\]
Since $|I_t|=2r+3=\gamma_t^1$, a union bound over the interface sources and the two graphs gives \cref{eq:pair-call}.
\end{proof}

\section{Adaptive Arbitrary-Decoder Barrier}\label{sec:oracle-barrier}

\begin{theorem}[Adaptive scalar-oracle exactification barrier]\label{thm:oracle-barrier}
Let $\mathcal A$ be any adaptive scalar-oracle exactification from \cref{def:oracle}, run on both $G_t^0$ and $G_t^1$ with the same identifier assignment. If
\[
 \Prb[\mathcal A(G_t^b)=\gamma_t^b]\ge\frac23
 \qquad\text{for }b\in\{0,1\},
\]
then
\begin{equation}\label{eq:oracle-work}
 \min_{b\in\{0,1\}}\E[W(G_t^b)]
 \ge \frac{n_t}{18\gamma_t^1}
 =\Omega\!\left(\frac{n_t}{\log n_t}\right).
\end{equation}
\end{theorem}
\begin{proof}
Couple the two executions using the same controller randomness, and, while their exposed transcripts agree, the same fresh uniform source set in each call. Let $\mathcal D$ be the event that the exposed transcripts ever diverge. By Definition~\ref{def:oracle}, the two executions begin with identical
controller-visible information. Before the first divergence, their complete
exposed transcripts and coupled private randomness therefore agree, so the
controller has the same state on both graphs and makes the same stop/continue
decision and chooses the same $Q_j,k_j$.

Let $L_j$ be the event that the transcripts agree through call $j-1$ and both executions continue to call $j$, and let $\mathcal D_j$ be the event that the first divergence occurs at call $j$. On $L_j$, \cref{lem:pair-call} applies conditionally on the complete pre-sampling state. Writing $B_j=\min\{Q_j,k_j\}$ on $L_j$ and $B_j=0$ otherwise,
\[
 \Prb[\mathcal D_j\mid\mathcal G_j]
 \le \ind_{L_j}\frac{6\gamma_t^1 B_j}{n_t}.
\]
The events $\mathcal D_j$ are disjoint. Conditional expectation and Tonelli's theorem give
\begin{equation}\label{eq:divergence}
 \Prb[\mathcal D]
 \le \frac{6\gamma_t^1}{n_t}
 \E\!\left[\sum_{j\ge1}\ind_{L_j}B_j\right]
 =\frac{6\gamma_t^1}{n_t}\E[W_*],
\end{equation}
where $W_*$ is the workload accumulated along the common transcript. Pathwise, $W_*\le W(G_t^b)$ for both $b$.

On $\mathcal D^c$, the complete exposed transcripts and private coins are identical, so the arbitrary final decoder returns the same value on both inputs. Since $\gamma_t^0\ne\gamma_t^1$, at most one of the two executions can then be correct. Hence, with $\mathrm{Succ}_b=\{\mathcal A(G_t^b)=\gamma_t^b\}$,
\[
 \Prb[\mathrm{Succ}_0]+\Prb[\mathrm{Succ}_1]
 \le 1+\Prb[\mathcal D].
\]
The assumed $2/3$ success on both graphs yields $\Prb[\mathcal D]\ge1/3$. Combining this with \cref{eq:divergence} gives $\E[W_*]\ge n_t/(18\gamma_t^1)$, and therefore the same lower bound holds for both expected workloads.
\end{proof}

\begin{corollary}[Standard packetized realization]\label{cor:oracle-rounds}
If the calls of \cref{thm:oracle-barrier} are executed sequentially by the standard packetized estimator, then exactness probability at least $2/3$ on both $G_t^0,G_t^1$ implies
\[
 \min_{b\in\{0,1\}}\E[R(G_t^b)]
 =\Omega(n_t/\log n_t).
\]
\end{corollary}
\begin{proof}
Each call serializes $B_j$ retained $O(\log n_t)$-bit records over every edge, so $R\ge cW$ pathwise for an absolute constant $c>0$. Apply \cref{thm:oracle-barrier}.
\end{proof}

\section{Bridgeless Robustness for Direct Retuning}\label{sec:bridgeless}

Fix $t\ge1$ and set
\[
 g_t=8t+1,\qquad \ell_t=8t+3=g_t+2,\qquad m_t=2^t-1.
\]
Start with a central odd cycle $C_t=(c_0,\ldots,c_{g_t-1})$. At every $c_i$, attach two vertex-disjoint perfect rooted binary trees of height $t-1$. We call the tree edges together with the anchor-root edges the \emph{skeleton links}. For every skeleton link $\{u,v\}$, keep the direct edge and add an internally vertex-disjoint $u$--$v$ path of length $8t+2$. Each skeleton link therefore lies on an odd gadget cycle of length $\ell_t$. Denote the resulting graph by $\widehat H_t$ and put $\widehat n_t=|V(\widehat H_t)|$.

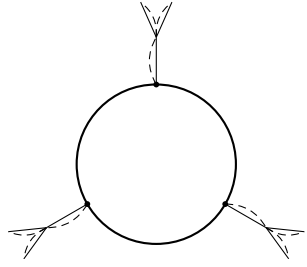
\begin{figure}[H]
\centering
\begin{tikzpicture}[scale=.78,every node/.style={font=\scriptsize}]
\draw[thick] (0,0) circle (1.35);
\foreach \a in {90,210,330}{
  \coordinate (c) at (\a:1.35);
  \fill (c) circle (1.4pt);
  \coordinate (u) at ($(c)+(\a:0.8)$);
  \coordinate (v1) at ($(u)+(\a-24:0.65)$);
  \coordinate (v2) at ($(u)+(\a+24:0.65)$);
  \draw (c)--(u); \draw (u)--(v1); \draw (u)--(v2);
  \draw[densely dashed] (c) to[bend left=30] (u);
  \draw[densely dashed] (u) to[bend left=28] (v1);
  \draw[densely dashed] (u) to[bend right=28] (v2);
}
\node[align=center] at (0,-2.3) {each skeleton link keeps its direct edge\\and receives a long detour, forming an odd cycle};
\end{tikzpicture}
\caption{Schematic of $\widehat H_t$. The direct skeleton preserves short distances, while every skeleton link belongs to a gadget cycle.}\label{fig:bridgeless}
\end{figure}

\begin{lemma}[Bridgeless family structure]\label{lem:structure}
The graph $\widehat H_t$ is a connected odd cactus with
\begin{align}
 \widehat n_t&=(8t+1)\bigl[1+2(2^t-1)(8t+2)\bigr]=\Theta(t^2 2^t),\label{eq:n}\\
 2&\le\delta(\widehat H_t)\le\Delta(\widehat H_t)\le6.\label{eq:degree}
\end{align}
Every edge lies on a cycle, so $\core(\widehat H_t)=\widehat H_t$ and $\Bridges(\widehat H_t)=\varnothing$. Its unique shortest cycle is $C_t$, with
\[
 g(\widehat H_t)=g_t=\Theta(\log\widehat n_t),
 \qquad
 D(\widehat H_t)=\Theta(\log\widehat n_t).
\]
\end{lemma}
\begin{proof}
Each central anchor owns two binary scaffolds, each with $m_t$ skeleton vertices and $m_t$ skeleton links. Every link receives $8t+1$ new detour vertices, giving \cref{eq:n} and $\log_2\widehat n_t=t+O(\log t)$. Central and non-leaf skeleton vertices have degree at most six; leaves and detour vertices have degree two.

Ignoring the detours leaves $C_t$ with trees attached. Adding the internally disjoint detour for a skeleton link creates exactly one gadget cycle around that link. Hence the simple cycles are precisely $C_t$ and the gadget cycles, all edge-disjoint. Every edge lies on one of them. Since $|C_t|=8t+1$ and every gadget cycle has length $8t+3$, $C_t$ is the unique shortest cycle. Every vertex is within $5t+1$ of $C_t$, giving diameter at most $14t+2$, while two central vertices at cycle distance $4t$ give diameter at least $4t$.
\end{proof}

By \cref{lem:block-witness}, a gadget block contributes at least $g_t+2$, while the central block contributes
\begin{equation}\label{eq:central-candidate}
 g_t+2d(s,C_t).
\end{equation}
Thus $\Estimate(\widehat H_t,S,k)=g_t$ exactly when some central source $s=c_i$ is sampled and retained at both endpoints of its antipodal central edge.

\begin{lemma}[Linear antipodal crowding]\label{lem:crowding}
At either endpoint $x$ of the antipodal central edge of any source $s=c_i\in C_t$, at least
\[
 N_t\ge8t(2t-1)(2^t-1)\ge\frac{\widehat n_t}{25}
\]
vertices are strictly closer to $x$ than $s$ is.
\end{lemma}
\begin{proof}
Here $d(x,s)=4t$. Consider the $2t-1$ central anchors within cycle distance $t-1$ of $x$. Every skeleton vertex in their two scaffolds lies within distance at most $2t-1$ of $x$. From every skeleton link, take the first $2t$ internal detour vertices from each endpoint. These sets are disjoint and every selected vertex $z$ satisfies
\[
 d(x,z)\le(2t-1)+2t=4t-1<d(x,s).
\]
There are $2(2^t-1)$ skeleton links per anchor, giving the first bound. From \cref{eq:n}, $\widehat n_t\le198t^2(2^t-1)$ while $N_t\ge8t^2(2^t-1)$, hence $N_t/\widehat n_t\ge4/99>1/25$.
\end{proof}

\begin{lemma}[Per-call exactness on $\widehat H_t$]\label{lem:per-call}
For a fresh exchangeable call on $\widehat H_t$, conditional on arbitrary $Q,k$ and with $B=\min\{Q,k\}$,
\begin{equation}\label{eq:per-call}
 \Prb[\Estimate(\widehat H_t,S,k)=g_t\mid Q,k]
 \le \min\left\{1,\frac{25g_tB}{\widehat n_t}\right\}.
\end{equation}
\end{lemma}
\begin{proof}
Exactness requires a central source $s=c_i$ to be retained at an antipodal endpoint. By \cref{lem:crowding,lem:rank-bound}, for each $s$ this has probability at most $25B/\widehat n_t$. A union bound over the $g_t$ central sources gives \cref{eq:per-call}.
\end{proof}

\begin{theorem}[Bridgeless direct-retuning barrier]\label{thm:bridgeless}
Let $\mathcal A$ be any direct-retuning controller from \cref{def:direct} on $\widehat H_t$. Then
\begin{equation}\label{eq:bridge-main}
 \Prb[\mathcal A(\widehat H_t)=g_t]
 \le \frac{25g_t}{\widehat n_t}\E[W].
\end{equation}
Consequently, success probability at least $2/3$ requires
\[
 \E[W]=\Omega\!\left(\frac{\widehat n_t}{g_t}\right)
 =\Omega\!\left(\frac{\widehat n_t}{\log\widehat n_t}\right).
\]
The same asymptotic bound holds for expected rounds in the standard sequential packetized realization.
\end{theorem}
\begin{proof}
Let $L_j=\{T\ge j\}$ and let $E_j$ be the event that executed call $j$ returns $g_t$. Conditional on the complete pre-sampling state, the current source set is fresh and uniform, so \cref{lem:per-call} gives
\[
 \Prb[E_j\mid\mathcal G_j]
 \le \ind_{L_j}\frac{25g_tB_j}{\widehat n_t}.
\]
The final minimum equals $g_t$ only if some $E_j$ occurs. Summing over calls and applying conditional expectation and Tonelli's theorem gives \cref{eq:bridge-main}. The workload conclusion follows immediately. For the standard packetized realization, $R\ge cW$ pathwise for an absolute constant $c>0$.
\end{proof}

\section{Matching Source-Hitting Scale Under a Girth Promise}\label{sec:source-hit}

\begin{proposition}[Source-hitting upper bound]\label{prop:source-hit}
Suppose an $n$-vertex undirected unweighted graph $G$ is promised to satisfy $g(G)\ge h$, where $h\ge1$ is known. There is a one-call nearest-source algorithm that computes exact girth with probability at least $2/3$ using retained-source workload $O(n/h)$ and $O(n/h+D)$ CONGEST rounds in the standard packetized realization. With high probability, the corresponding bounds are $O((n/h)\log n)$ workload and $\wtO(n/h+D)$ rounds.
\end{proposition}
\begin{proof}
If $G$ is acyclic, the estimator returns $+\infty$. Otherwise fix a shortest cycle $C$ of length $g\ge h$. Choose exactly
\[
 q=\min\{n,\lceil cn/h\rceil\}
\]
sources uniformly without replacement, for a sufficiently large absolute constant $c$, and run $\Estimate(G,S,q)$ with capacity $k=q$. If $q<n$,
\[
 \Prb[S\cap V(C)=\varnothing]
 =\frac{\binom{n-g}{q}}{\binom nq}
 \le\left(1-\frac gn\right)^q
 \le e^{-qg/n}
 \le e^{-c}.
\]
Choose $c$ so that this is at most $1/3$. A shortest cycle is isometric: for $u,v\in C$, a path shorter than the shorter $C$-arc would combine with that arc to contain a cycle shorter than $C$. Hence, for any sampled $s\in C$, the BFS distances on $C$ equal the cycle distances. All vertices at cycle distance strictly below $g/2$ have their unique shortest path along the corresponding shorter arc of $C$. If $g=2r+1$ is odd, the two BFS branches meet across the unique edge whose endpoints are both at distance $r$ from $s$. If $g=2r$ is even, let $v$ be the antipodal cycle vertex at distance $r$ from $s$; exactly one of its two cycle neighbors is its BFS parent, and the other incident cycle edge has endpoint distances $r$ and $r-1$. Thus in either parity there is a non-parent cycle edge $\{x,y\}$ satisfying
\[
 d(s,x)+d(s,y)+1=g.
\]
Since $k=q$, every sampled source is retained everywhere, so this edge contributes $g$. The estimator never returns below girth, and therefore returns exactly $g$. Source detection and table exchange take $O(q+D)=O(n/h+D)$ rounds~\cite{LenzenPeleg2013}. Taking $q=\min\{n,\lceil c(n/h)\log n\rceil\}$ gives the high-probability statement.
\end{proof}

For $\widehat H_t$, setting $h=g_t$ yields $O(\widehat n_t/g_t)$ workload, matching \cref{thm:bridgeless} up to constants. The same $n/g$ dependence is therefore both necessary and sufficient for the direct source-hitting architecture on that family.

\section{Discussion}\label{sec:discussion}

\paragraph{Two complementary robustness axes.}
The pair $G_t^0,G_t^1$ establishes information-interface robustness: sampled source identities are visible, the exposed transcript may persist across calls, all parameters and stopping are adaptive, and the final decoder is arbitrary. The family $\widehat H_t$ establishes graph-structural robustness for direct retuning: minimum degree is at least two, every edge lies on a cycle, the graph is its own $2$-core, and the diameter is logarithmic. Together they separate the role of the scalar estimator interface from immediate graph preprocessing.

\paragraph{Why approximation retains its advantage.}
For an odd cycle block $Z$, \cref{eq:block-candidate} is $|Z|+2d(s,Z)$. Approximation can absorb a positive source-to-cycle distance, while exactness identifies the zero-distance event or, in the paired construction, must distinguish two instances for which the bulk-source values are exactly equal. This is the geometric source of the exactness barrier.

\paragraph{Information exposed by the oracle.}
The scalar-oracle theorem preserves fresh exchangeable source selection and exposes $(S_j,M_j)$ from every call, while nearest-source tables remain encapsulated. The proof therefore identifies richer source selection and exported or jointly compressed table information as the next design dimensions for exact algorithms based on nearest-source exploration.

\section{Conclusion}

Adaptive composition substantially broadens what can be done with the nearest-source girth estimator, but the scalar information it exposes still has an exactness threshold. A same-size length-transfer pair makes every bulk source behave identically across two different girths, and a rank argument shows that separating the pair requires $\Omega(n/\log n)$ expected retained-source workload even with visible source identities, persistent state, adaptive stopping, and arbitrary final decoding. A complementary bridgeless family yields the same scale for direct retuning on graphs equal to their own $2$-core. The matching source-hitting proposition shows that the $n/g$ dependence captured by these barriers is the natural scale for direct exact certification.


\bibliography{references}

\appendix
\section{Finite Verification}\label{app:verification}
The supplementary verification script checks both constructions on finite instances. For the scalar-indistinguishability pair it verifies equal vertex counts, the two claimed girths, the unique cycle witness for every source, equality of bulk-source candidates, equality of the exact predecessor sets used for nearest-source retention, and the interface-crowding bound; it also performs coupled random-call checks of the scalar-invariance event. For $\widehat H_t$ it verifies the cycle partition, degree bounds, absence of bridges, unique-shortest-path and witness identities, the exact-certifying-source characterization, and antipodal crowding. These checks are supplementary to the analytic proofs.

\end{document}